\documentclass{article}
\usepackage{spconf,amssymb,amsmath,graphicx,bm,pifont}
\usepackage{tikz,multirow,pgfplots,subcaption,setspace}
\usepackage{amsthm}
\usepackage{enumitem}
\usepackage{cite}
\title{contrastive separative coding for Self-supervised representation learning}
\name{Jun Wang$^{\star}$ \qquad Max W. Y. Lam$^{\star}$ \qquad Dan Su$^{\star}$ \qquad Dong Yu$^{\dagger}$}
\address{$^{\star}$ Tencent AI Lab, Shenzhen, China\\$^{\dagger}$ Tencent AI Lab, Bellevue WA, USA} 
\newtheorem{definition}{Definition}[section]
\newtheorem{assumption}{Assumption}[section]
\newtheorem{claim}{Claim}[section]

\begin{document}
\ninept
\maketitle
\begin{abstract}
To extract robust deep representations from long sequential modeling of speech data, we propose a self-supervised learning approach, namely Contrastive Separative Coding (CSC). Our key finding is to learn such representations by separating the target signal from contrastive interfering signals. First, a multi-task separative encoder is built to extract shared separable and discriminative embedding; secondly, we propose a powerful cross-attention mechanism performed over speaker representations across various interfering conditions, allowing the model to focus on and globally aggregate the most critical information to answer the ``query" (current bottom-up embedding) while paying less attention to interfering, noisy, or irrelevant parts; lastly, we form a new probabilistic contrastive loss which estimates and maximizes the mutual information between the representations and the global speaker vector. While most prior unsupervised methods have focused on predicting the future, neighboring, or missing samples, we take a different perspective of predicting the interfered samples. Moreover, our contrastive separative loss is free from negative sampling. The experiment demonstrates that our approach can learn useful representations achieving a strong speaker verification performance in adverse conditions.
\end{abstract}

\section{Introduction}
\vspace{-0.5em}
\label{sec:intro}
Learning high-level representations from labeled data has achieved marvelous successes in modern speech processing. To extract reliable speaker representations in realistic situations, however, conventional speaker verification, speaker identification, and speaker diarization (SV, SI, SD) systems generally require complicated pipelines \cite{David2018xvector, liu2018sv, Yoshua2018sv}. One has to prepare three independent models: (i) a speech activity detection model to generate short speech segments with no interference or overlapping, (ii) a speaker embedding extraction model, and (iii) a clustering or a PLDA model including covariance matrices is needed to group the short segments to the same or different speaker. Jointly supervised modeling methods\cite{nara2019sd, huang2020sd, mia2018sd} have been studied to alleviate the long preparation process and take into account the dependencies between these models. More recently, end-to-end neural speaker diarization \cite{yusuke2019sd_pit,yusuke2019sd_attn, yusuke2020sd_attractors} has been proposed to overcome the situation that the previous systems can not deal with speaker overlap parts because each time slot is assigned to one speaker.

Despite the breakthrough seen by these supervised methods, many challenges remain, such as data robustness, efficiency, and generalization. To alleviate these issues, improving representation learning requires features that are less specialized towards solving a single supervised task. Unsupervised or self-supervised learning \cite{wav2vec, dim19, nce18, word2vec, speech2vec} is a promising apparatus towards generic and robust representation learning. 
A common strategy is to use the conditional dependency between the features of interest and the same shared high-level latent information/context. Advanced work in unsupervised or self-supervised learning has successfully used the strategy to learn deep representations by predicting neighboring or missing words \cite{word_embedding13, Jacob2019Bert}, predicting the relative position of image patches \cite{carl2015visual} or color from grey-scale \cite{richard2016color}, or most recently predicting future frames \cite{nce18} or contextual information \cite{dim19}.

Classical predictive coding theories in neuroscience suggest that the brain predicts observations at various levels of abstraction \cite{karl2005neuro,rao1999neuro}. When one is listening to overlapped speakers, we infer the features of interest conditionally dependent on both low-mid levels of abstraction of the same speaker (e.g., spectrum continuity and structure, timbre, phoneme, syllable, etc.) and high levels of abstraction (E.g., speaker characterizations, spatial position, etc.). Inspired by the above, we hypothesize these different levels of abstraction have shared bottom features. We propose Contrastive Separative Coding (CSC) model as the following: first, an encoder compresses the raw input into a compact latent space of separable and discriminative embedding shared by various levels of abstraction task (Sec.\ref{sec:2}, \ref{sec:pit}); secondly, we propose a powerful cross-attention model in this latent space to model the high-level abstraction of speaker representations (Sec.\ref{sec:cross}); lastly, we form a new probabilistic contrastive loss which estimates and maximizes the mutual information between the representations and the global speaker vector (Sec.\ref{sec:csc}).
\section{Related Work and our contributions}
\vspace{-0.5em}
\label{sec:related}
The compositional attention networks \cite{Hudson2018ComposAttn} decomposes machine reasoning into a series of attention-based reasoning operations that are directly inferred from the data, without resorting to any strong supervision. Interaction between its two modalities – visual and textual, or knowledge base and query – is mediated through soft-attention only. The soft attention enhances its model’s ability to perform reasoning operations that pertain to the global aggregation of information across different regions of the image. Earlier, cross-stitch strategy \cite{ishan2016cross} has been proposed between text and image, on which both \cite{Hudson2018ComposAttn} and our proposed cross-attention approach can be regarded as new variations. To the best of our knowledge, however, we are the first to use a cross-attention strategy between high-level speaker representations and low-level speech features.
\begin{figure*}[t!]
 \centering
    \includegraphics[width=\linewidth]{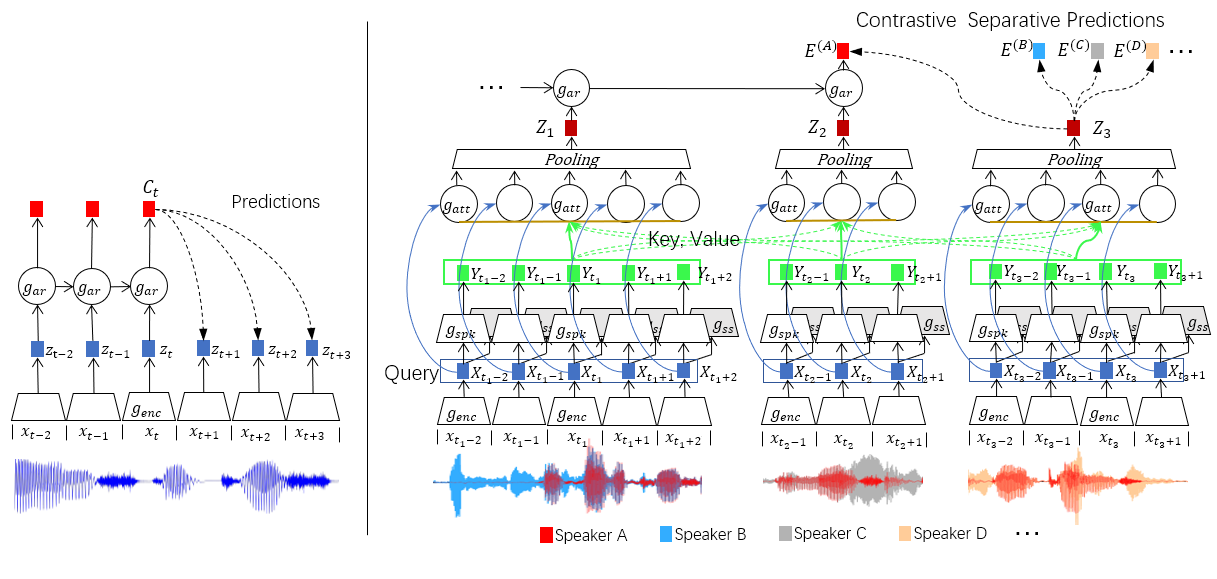}
\vspace{-3em}
\caption{Left: overview of Contrastive Predictive Coding \cite{nce18}. Right: our proposed Contrastive Separative Coding.}
\label{fig:archi}
\vspace{-2em}
\end{figure*}

Another closely related prior work that does not resort to any strong supervision is the Contrastive Predictive Coding (CPC) \cite{nce18}, as shown on the left in Fig. \ref{fig:archi}. Comparing it to our proposed CSC on the right, we summarize our key contributions in four folds:
\begin{itemize}[leftmargin=*]
\item While most unsupervised prior work, including \cite{nce18}, has focused on predicting the future, neighboring, or missing samples, our perspective is different in that it focuses on interfered samples, i.e., learning separative representations from observations with various interference.
\item To learn mid-level representations by using the powerful cross-attention strategy.
\item To study the prediction strategy on latent global speaker features, for the first time at a much higher level of abstraction than the prior work \cite{wav2vec,dim19, nce18,speech2vec,Jacob2019Bert}. We consider predicting the ``global feature" that spans various observations more interesting, especially across various interfering conditions, thus the model needs to infer more global structure to preserve the latent shared information. 
\item A novel contrastive separative loss is first used as a bound on mutual information. In contrast to the prior common strategy, our proposed loss models an inverse conditional dependency between the features of interest and the shared high-level latent context, which also holds for estimating the mutual information thanks to its symmetric property. This makes our conditional predictions easier to model and free from negative sampling as generally required in prior work, at the cost of weak supervision about the shared high-level latent context -- here is to know the set of mixtures containing signals by the same speaker, e.g., the ones marked in red in Fig. \ref{fig:archi}.
\vspace{-1.5em}
\end{itemize}
\section{Proposed Approach}
\vspace{-0.5em}
\label{sec:2}
On the right of Fig. \ref{fig:archi} is the architecture of Contrastive Separative Coding (CSC) model. First, an encoder $g_{enc}$ maps the encoded and segmented input mixture sequences of $x_{t_i}, (i=1,2,3,...)$ to sequences of latent representations $X_{t_i}$. We denote them as 2D tensors $\textbf{X}_i \in \mathbb{R}^{D\times S_i}$, where $D$ is the feature dimension, $S_i$ is the sequence length. Secondly, task-specific encoders, $g_{spk}$ and $g_{ss}$, map the shared latent representations $\textbf{X}_i$ to compact latent spaces of various levels of abstraction tasks, i.e., speaker embedding ($\textbf{Y}_i \in \mathbb{R}^{D\times S_i}$) and speech embedding, respectively. Next, a powerful cross-attention model $g_{att}$ in the speaker embedding space computes the high-level abstraction of speaker representations that are pooled into $\textbf{Z}_i\in \mathbb{R}^{D}$. Finally, an autoregressive model $g_{ar}$ summarizes all the past $\textbf{Z}_i$ to a global speaker vector $\textbf{E} = g_{ar}(\textbf{Z}_{i\le{t}})\in \mathbb{R}^{D}$.
\subsection{Bottom-up Cross Attention}
\vspace{-0.5em}
\label{sec:cross}
We cast the query vector onto the latent speaker space via the cross-attention model $\mathbf{Z}_i=\text{Pooling}(g_{att}(\textbf{X}_i, \textbf{Y}_j))$. The bottom-up queries are raised from the shared bottom features $\textbf{X}$ to retrieve the most relevant speaker-dependent global characteristics from the latent speaker space of $\textbf{Y}$ and filter out non-salient, noisy, or redundant information. $\text{Query}(\cdot)$, $\text{Key}(\cdot)$, and $\text{Value}(\cdot)$ denote linear transformation functions, where the corresponding input vectors are linearly projected along the feature dimension ($\mathbb{R}^{D}\mapsto \mathbb{R}^{D}$) into query, key, and value vectors, respectively. Our detailed implementation of these functions is the same as \cite{matt2018self}.

Specifically, this is achieved by computing the inner product between $\text{Query}(\textbf{X}_i)\in\mathbb{R}^{D\times S_i}$ and $\text{Key}(\mathbf{Y}_j)\in\mathbb{R}^{D\times S_j}$:
\begin{equation}
\label{eq:a1}
   \mathbf{a}_{i,j}=\text{softmax}(\text{Query}(\textbf{X}_i)^{\top}\cdot\text{Key}(\mathbf{Y}_j)).
\end{equation}
yielding a cross-attention distribution $\mathbf{a}_{i,j}\in\mathbb{R}^{S_i\times S_j}$, where $S_i, S_j$ are the sequence lengths of either different observations or the same observation ($S_i= S_j$, amount to self-attention).

Finally, we compute the sum of $\text{Value}(\mathbf{Y}_j)\in\mathbb{R}^{D\times S_j}$ weighted by the cross-attention distribution $\mathbf{a}_{i,j}$, and then average over $S_i$ segments, to produce a separative embedding $\mathbf{Z}_{i}\in \mathbb{R}^{D}$:
\begin{equation}
\label{eq:a2}
    \mathbf{Z}_i=\frac{1}{S_i}\sum_{S_i}\sum_{S_j}\mathbf{a}_{i,j}\cdot\text{Value}(\mathbf{Y}_j)^\top.
\end{equation}
\subsection{Contrastive Separative Coding Loss}
\vspace{-0.5em}
\label{sec:csc}
Next, we introduce CSC loss, which takes the following form:
\begin{align}
\label{eq:1.1}
    &\mathcal{L}_{\tiny\text{CSC}} = -
    \mathbb{E}_{\mathcal{D}}\left[\log\left({f(\mathbf{Z}^{(n_c)}, \textbf{E}^{(n_c)})}/{\sum_{n=1}^{N}f(\mathbf{Z}^{(n_c)}, \textbf{E}^{(n)})}\right)\right],\\
\label{eq:1.2}
    &f(\mathbf{Z}^{(n_c)}, \textbf{E}^{(n)})= \exp\left(-\alpha\lVert\mathbf{Z}^{(n_c)}- \textbf{E}^{(n)}\rVert_{2}^{2}\right),
\end{align}
We denote $\textbf{E}^{(n)}$ as the global speaker vector for speaker $n$ out of the overall $N$ speakers in the training set $\mathcal{D}$ and the separative embedding $\mathbf{Z}^{(n_c)}$ for the $c$-th separated signal, for $c=1,...,C$, given that each audio input is a mixture of $C$ sources and the $c$-th source is generated by one of the speakers indexed by $n_c$ out of the overall $N$ speakers. We assume that the training set $\mathcal{D}$ sufficiently defined the sample space of their joint distributions. For notational simplicity, we omit the subscript $i$ and use $\mathbf{Z}$ to denote $\mathbf{Z}^{(n_c)}$, $\textbf{E}$ to denote $\textbf{E}^{(n_c)}$ in the following text, unless otherwise specified.
First, we study the relationship between CSC loss and mutual information:
\begin{definition}
\label{def:1}
Mutual information of the global speaker vector and the separative embedding is defined as
\begin{align}
\label{eq:2}
\mathcal{I}(\textbf{E};\mathbf{Z})=\mathbb{E}_{\mathcal{D}}\left[\log \frac{p(\textbf{E}, \mathbf{Z})}{p(\textbf{E}) p(\mathbf{Z})}\right].
\end{align}
\end{definition}
Then, we make the following assumptions:
\begin{assumption}
\label{ass:1}
With a suitable mathematical form for function $f(\cdot)$, we can model a density ratio defined as
\begin{align}
\label{eq:3}
f(\mathbf{Z}, \textbf{E}) \propto \frac{p(\textbf{E}|\mathbf{Z})}{p(\textbf{E})}=\frac{p(\mathbf{Z}|\textbf{E})}{p(\mathbf{Z})}.
\end{align}
\end{assumption}

\begin{assumption}
\label{ass:1}
Considering the case of $n\neq n_c$, since the separative embedding $\mathbf{Z}^{(n_c)}$ does not belong to speaker $n$, $\textbf{E}^{(n)}$ should not be dependent on $\mathbf{Z}^{(n_c)}$. Therefore, it is sensible to assume
\begin{align}
\label{eq:4}
p(\textbf{E}^{(n)}|\mathbf{Z}^{(n_c)})=p(\textbf{E}^{(n)}),&\,\,\,\,\,\,\,\,\forall\,\,n\neq n_c.
\end{align}
\end{assumption}

Based on these assumptions, we can deduce the claims:
\begin{claim}
\label{claim:1}
Minimizing CSC loss results in maximizing mutual information between global speaker vector and the separative embedding, since CSC loss $\mathcal{L}_{\tiny\text{CSC}}$ serves as an upper bound of the negative mutual information $-\mathcal{I}(\textbf{E};\mathbf{Z})$.
\end{claim}
\begin{proof}
To prove this claim, we substitute Eq. (\ref{eq:3}) into Eq. (\ref{eq:1.1}) and obtain the following results:
\begin{align*}
    \mathcal{L}_{\tiny\text{CSC}} &= -
    \mathbb{E}_{\mathcal{D}}\left[\log\left(\frac{\frac{p(\textbf{E}|\mathbf{Z})}{p(\textbf{E})}}{\sum_{n=1}^{N}\frac{p(\textbf{E}^{(n)}|\mathbf{Z})}{p(\textbf{E}^{(n)})}}\right)\right]\\
    &=\mathbb{E}_{\mathcal{D}}\left[\log\left(1+\frac{\sum_{n\neq n_c}\frac{p(\textbf{E}^{(n)}|\mathbf{Z})}{p(\textbf{E}^{(n)})}}{\frac{p(\textbf{E}|\mathbf{Z})}{p(\textbf{E})}}\right)\right]\\
    &=\mathbb{E}_{\mathcal{D}}\left[\log\left(1+(N-1)\frac{p(\textbf{E})}{p(\textbf{E}|\mathbf{Z})}\right)\right]\\
    &\geq\mathbb{E}_{\mathcal{D}}\left[\log\left(N\frac{p(\textbf{E})}{p(\textbf{E}|\mathbf{Z})}\right)\right]
    =\log N-\mathbb{E}_{\mathcal{D}}\left[\log\left(\frac{p(\textbf{E}, \mathbf{Z})}{p(\textbf{E})p(\mathbf{Z})}\right)\right]\\
    &=\log N-\mathcal{I}(\textbf{E};\mathbf{Z})
\end{align*}
This proves that $\mathcal{L}_{\tiny\text{CSC}}$ is an upper bound of $ -\mathcal{I}(\textbf{E};\mathbf{Z})$.
\end{proof}
Next, we study our proposed form of $f(\mathbf{Z}, \textbf{E})$ and its associated properties when minimizing CSC loss.
\begin{claim}
\label{claim:2}
Applying our proposed form of $f(\mathbf{Z}, \textbf{E})$ to $\mathcal{L}_{\tiny\text{CSC}}$ corresponds to treating each global speaker vector $\textbf{E}$ as a cluster centroid (Gaussian mean) of different separative embedding vectors $\mathbf{Z}$ generated by the same speaker $i_j$ with a learnable parameter $\alpha > 0$ controlling the cluster size (Gaussian variance):
\begin{align}
\label{eq:6}
    p(\mathbf{Z}|\textbf{E})&=\mathcal{N}(\textbf{E}, (2\alpha)^{-1}\mathbf{I})\\
\label{eq:7}
    p(\mathbf{Z})&=\mathcal{N}(\mathbf{0}, (2\alpha)^{-1}\mathbf{I}).
\end{align}
\end{claim}
\begin{proof}
Considering a density ratio $\hat{f}$:
\begin{align*}
\hat{f}(\mathbf{Z}, \textbf{E}) &\propto \frac{\mathcal{N}(\textbf{E}, (2\alpha)^{-1}\mathbf{I})}{\mathcal{N}(\mathbf{0}, (2\alpha)^{-1}\mathbf{I})}
=\frac{\exp\left(-1/2\cdot(2\alpha)\lVert\mathbf{Z}-\textbf{E}\rVert_2^2\right)}{\exp\left(-1/2\cdot(2\alpha)\lVert\mathbf{Z}\rVert_2^2\right)}\\
&=\exp\left(-\alpha\lVert\mathbf{Z}-\textbf{E}\rVert_2^2+\alpha\lVert\mathbf{Z}\rVert_2^2\right)
\end{align*}
Now, we evaluate $\mathcal{L}_{\tiny\text{CSC}}$ with the above-defined $\hat{f}(\mathbf{Z}, \textbf{E})$:

\begin{align*}
    \mathcal{L}_{\tiny\text{CSC}} &= -\mathbb{E}_{\mathcal{D}}\left[\log\left(\frac{\exp\left(-\alpha\lVert\mathbf{Z}-\textbf{E}\rVert_2^2+\alpha\lVert\mathbf{Z}\rVert_2^2\right)}{\sum_{n=1}^{N}\exp\left(-\alpha\lVert\mathbf{Z}-\textbf{E}^{(n)}\rVert_2^2+\alpha\lVert\mathbf{Z}\rVert_2^2\right)}\right)\right]\\
    &= -\mathbb{E}_{\mathcal{D}}\left[\log\left(\frac{\exp\left(-\alpha\lVert\mathbf{Z}-\textbf{E}\rVert_2^2\right)}{\sum_{n=1}^{N}\exp\left(-\alpha\lVert\mathbf{Z}-\textbf{E}^{(n)}\rVert_2^2\right)}\right)\right]\\
    &= -\mathbb{E}_{\mathcal{D}}\left[\log\left(\frac{f(\mathbf{Z}, \textbf{E})}{\sum_{n=1}^{N}f(\mathbf{Z}, \textbf{E}^{(n)})}\right)\right]
\end{align*}
\vspace{-1.5em}
\end{proof}
\begin{claim}
\label{claim:3}
With our proposed form of $f(\mathbf{Z}, \textbf{E})$, minimizing CSC loss results in minimizing the distance between the separative embedding $\mathbf{Z}$ and the corresponding global speaker vector $\textbf{E}$ meanwhile maximizing the distance between other global speaker vectors $\{\textbf{E}^{(n)} |\forall n\neq n_c\}$.
\end{claim}
\begin{proof}
By substituting Eq. (\ref{eq:1.2}) into Eq. (\ref{eq:1.1}), we have
\begin{align*}
\mathcal{L}_{\tiny\text{CSC}} &= -
    \mathbb{E}_{\mathcal{D}}\left[\log\left(\frac{\exp\left(-\alpha\lVert\mathbf{Z}- \textbf{E}\rVert_{2}^{2}\right)}{\sum_{n=1}^{N}\exp\left(-\alpha\lVert\mathbf{Z}- \textbf{E}^{(n)}\rVert_{2}^{2}\right)}\right)\right]\\
    &= 
    \mathbb{E}_{\mathcal{D}}\left[\alpha\lVert\mathbf{Z}- \textbf{E}\rVert_{2}^{2}+\log\left(\sum_{n=1}^{N}\exp\left(-\alpha\lVert\mathbf{Z}- \textbf{E}^{(n)}\rVert_{2}^{2}\right)\right)\right],
\end{align*}
which consists of two terms: (1) the first term is a scaled Euclidean distance with a scalar $\alpha > 0$, which minimizes the Euclidean distance between any global speaker vector and its corresponding separative embedding. (2) the second term is a logarithmic sum of exponentials, with a negative sign on the Euclidean distance, which pulls the separative embedding away from all other global speaker vectors.
\vspace{-1em}
\end{proof}
Besides, we can also relate the proposed CSC loss with the existing prior work.
\begin{claim}
\label{claim:4}
CSC loss can be seen as a rescaled $l$-2 normalization of InfoNCE loss proposed in \cite{nce18}.
\end{claim}
\begin{proof}
The InfoNCE loss has a different form, $f_\text{InfoNCE}(\mathbf{Z}, \textbf{E})=\exp\left(\mathbf{Z}^\top \textbf{E}\right),$
than our proposed $f(\mathbf{Z}, \textbf{E})$, we get
\begin{align*}
f(\mathbf{Z}, \textbf{E})&=\exp\left(-\alpha\lVert\mathbf{Z}- \textbf{E}\rVert_{2}^{2}\right)
=\frac{\exp\left(2\alpha\mathbf{Z}^\top \textbf{E}\right)}{\exp\left(\alpha\lVert\mathbf{Z}\rVert_{2}^{2}+\alpha\lVert\textbf{E}\rVert_{2}^{2}\right)}\\
&=\frac{f_\text{InfoNCE}(\mathbf{Z}, \textbf{E})^{2\alpha}}{\exp\left(\alpha\lVert\mathbf{Z}\rVert_{2}^{2}+\alpha\lVert\textbf{E}\rVert_{2}^{2}\right)}
\end{align*}
\end{proof}
\vspace{-3em}
\begin{figure*}[hbt]
\label{fig:4plots} 
      \begin{subfigure}{0.29\textwidth}
      \centering
      \includegraphics[width=\linewidth]{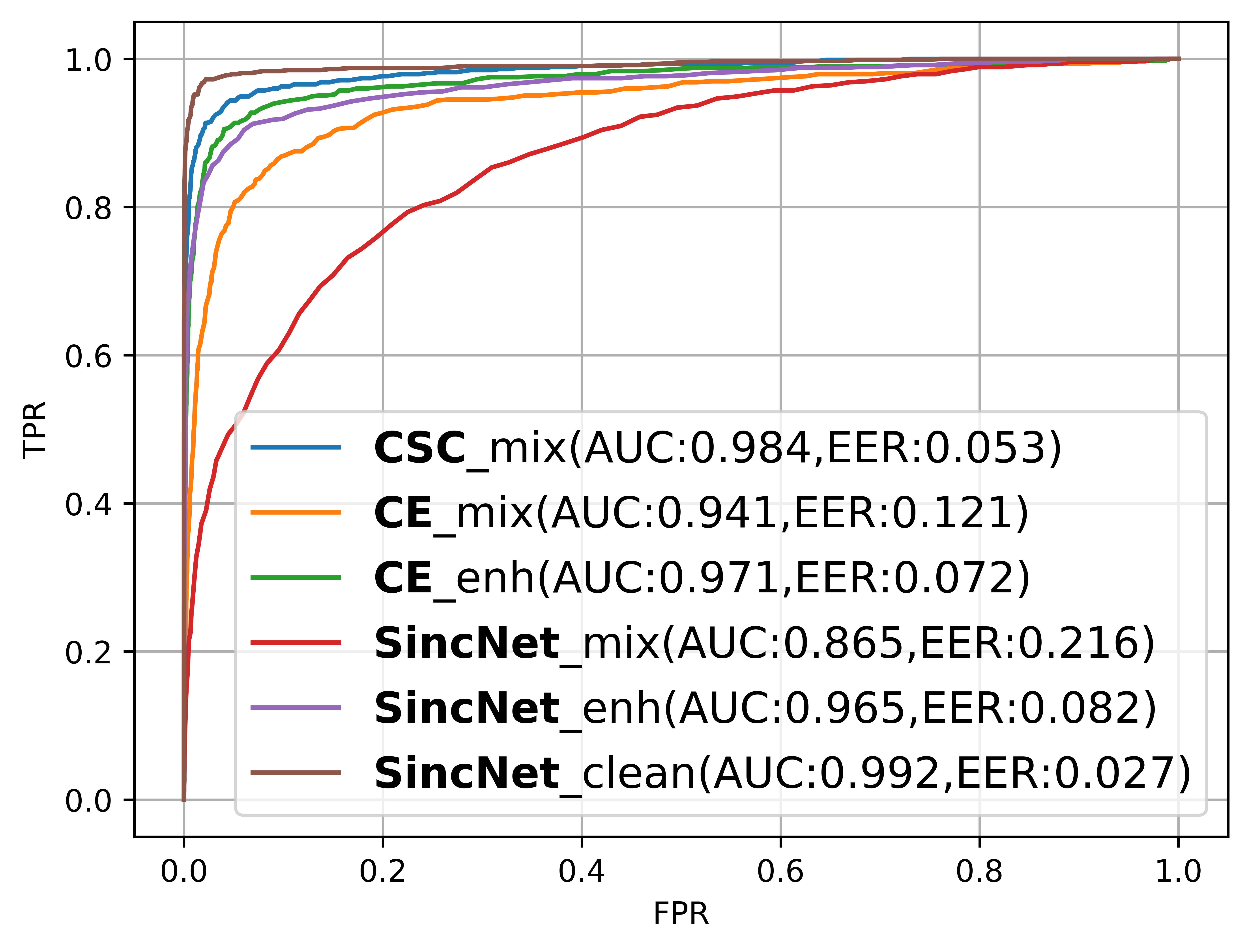}         
      \caption{ROCs}
      \vspace{-0.5em}
      \label{fig:roc}
      \end{subfigure}
      \begin{subfigure}{0.23\textwidth}
      \centering
      \includegraphics[width=\linewidth]{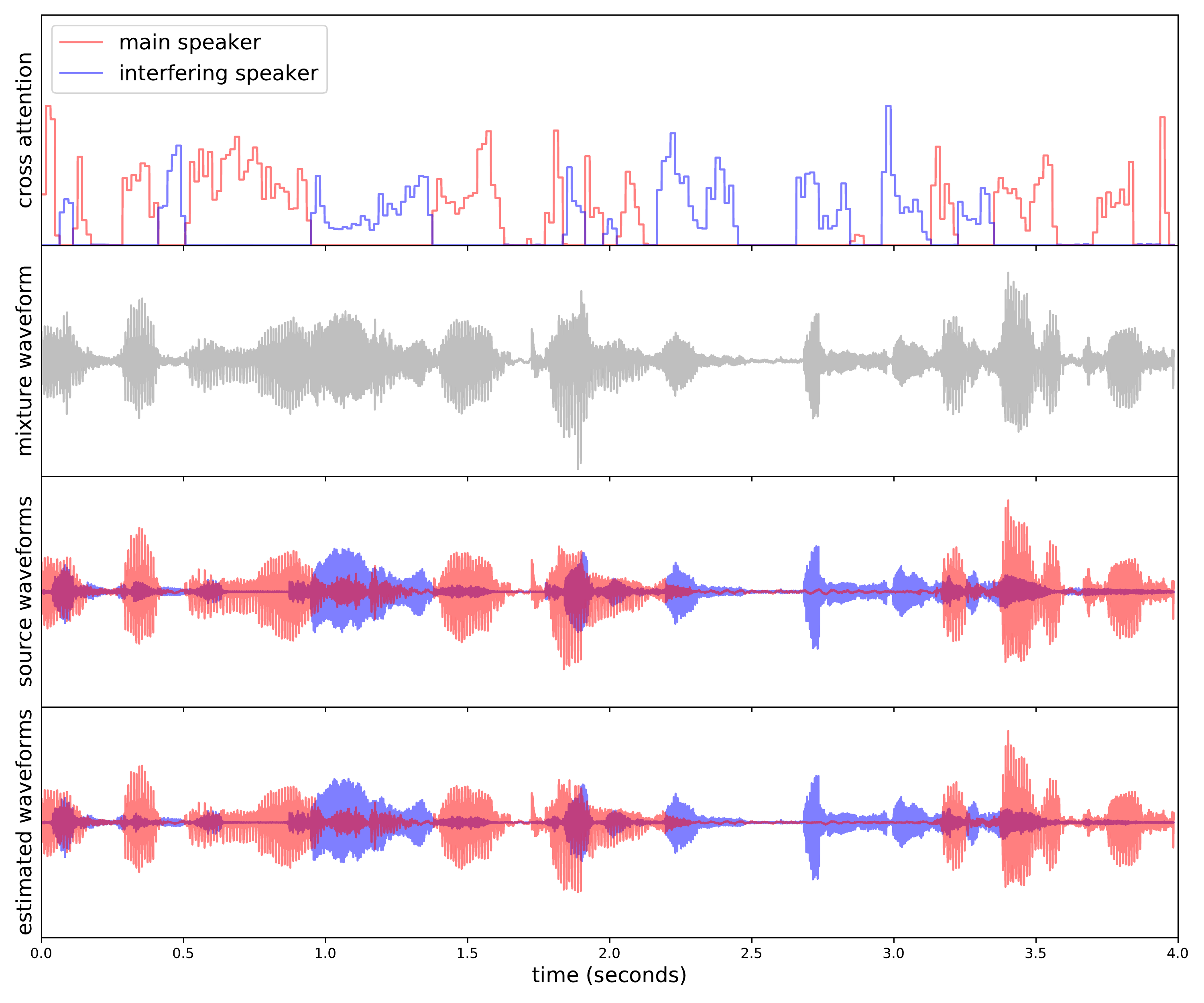}         
      \caption{A male interfered by a female}
      \vspace{-0.5em}
      \label{fig:male_female}
      \end{subfigure}
    \begin{subfigure}{0.23\textwidth}
      \centering
       \includegraphics[width=\linewidth]{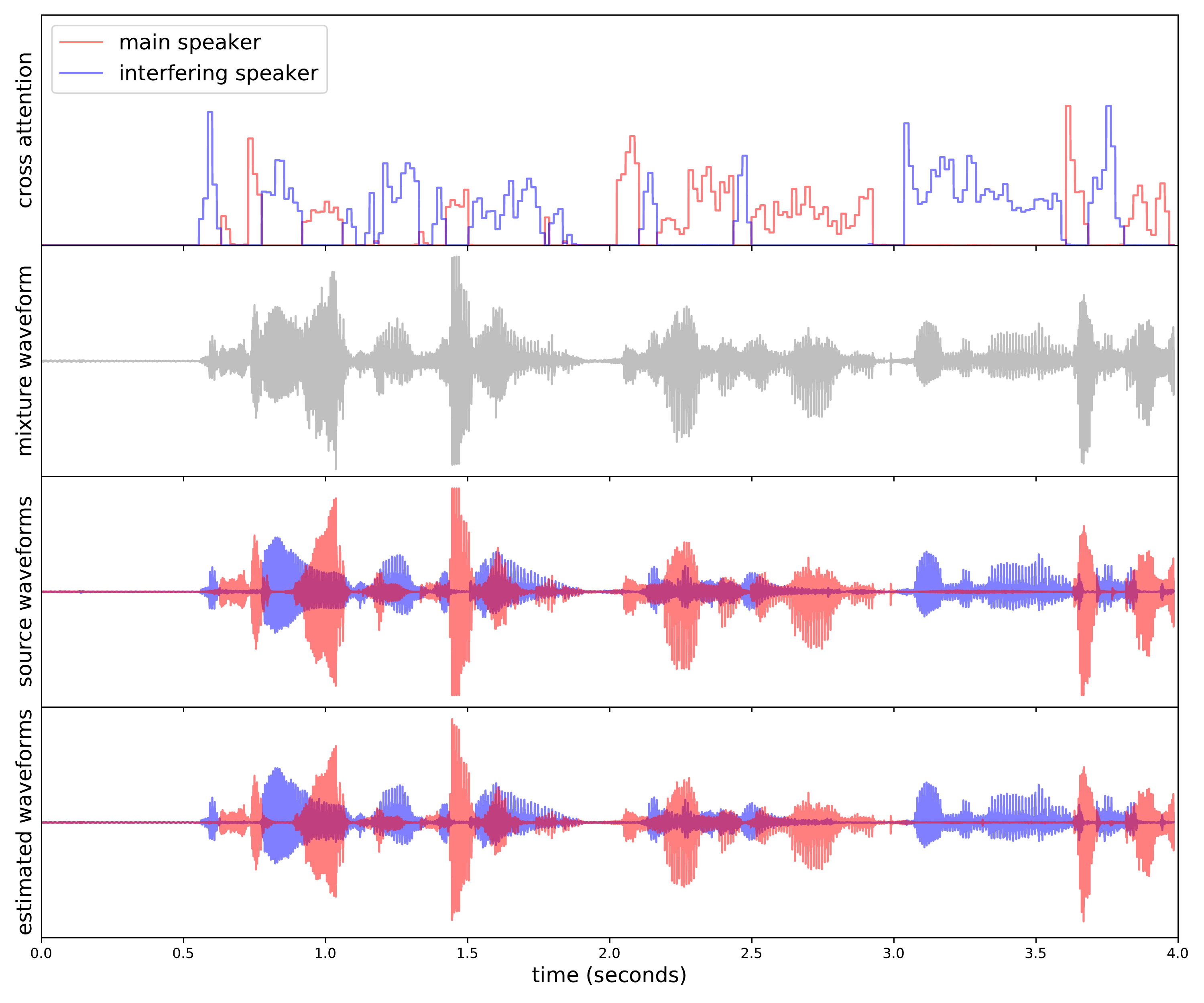} 
       \caption{A male interfered by another}
       \vspace{-0.5em}
       \label{fig:male_male}
       \end{subfigure}
      \begin{subfigure}{0.23\textwidth}
      \centering
      \includegraphics[width=\linewidth]{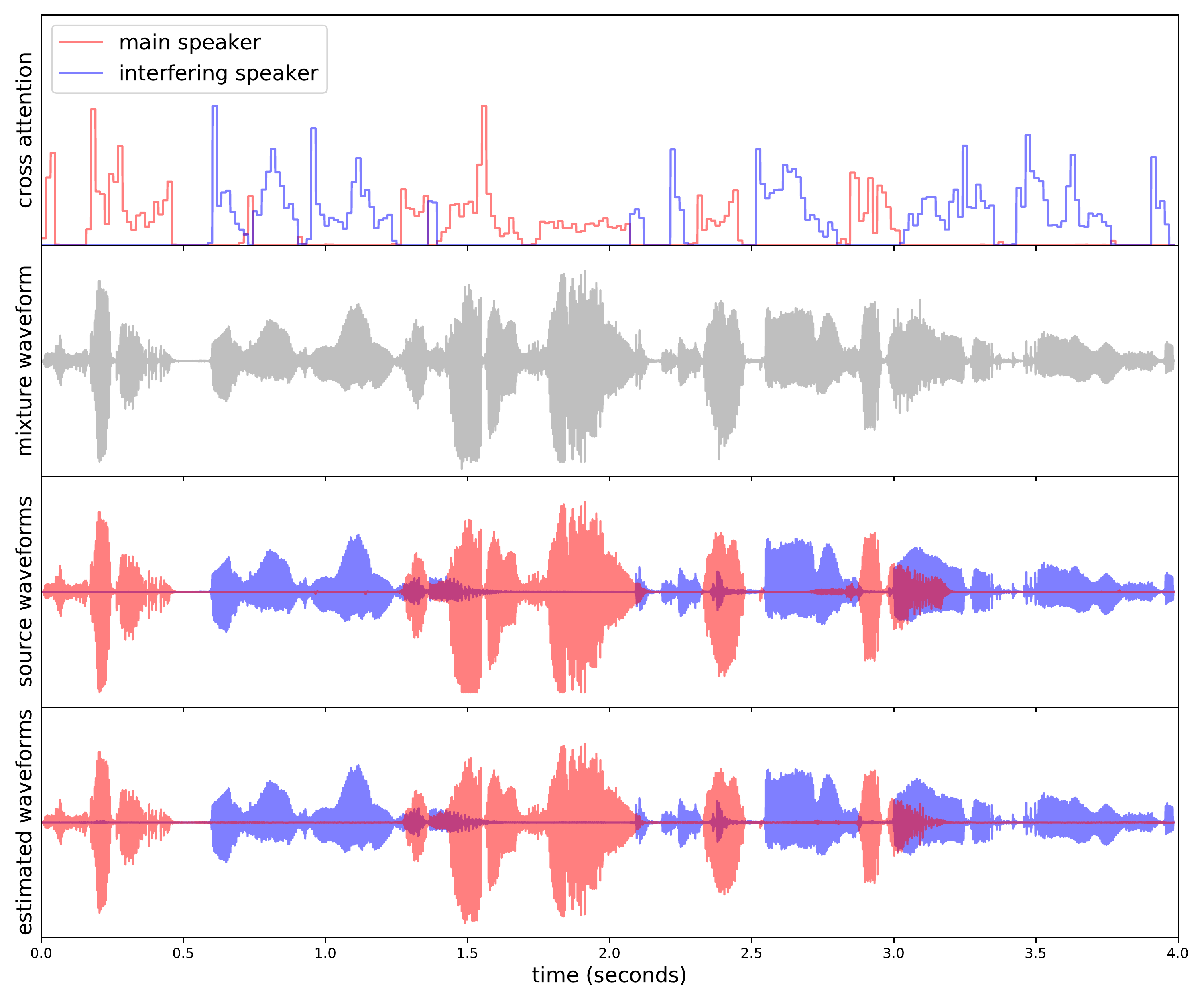} 
      \caption{A female interfered by another}
      \vspace{-0.5em}
      \label{fig:female_female}
      \end{subfigure}
\caption{(a) ROC curves by the various models, (b-d)Selective bottom-up cross attentions automatically learned based on mixture observations}
\vspace{-2em}
\end{figure*}
\subsection{Permutation Invariant Training Speedup}
\vspace{-0.5em}
\label{sec:pit}
As mentioned in Sec.\ref{sec:2}, the shared bottom features are jointly learned towards another level of abstraction tasks, i.e., speech separation (SS) to reconstruct the source signals. The joint training loss is $\mathcal{L}_{\tiny\text{SI-SNR}}+\lambda(\mathcal{L}_{\tiny\text{CSC}}+\mathcal{L}_{\tiny\text{reg}})$, where $\mathcal{L}_{\tiny\text{SI-SNR}}$ is the scale-invariant signal-to-noise ratio \cite{le2019sdr} loss for training $g_{ss}$, while $\mathcal{L}_{\tiny\text{CSC}}+\mathcal{L}_{\tiny\text{reg}}$ is for training $g_{spk}$, where $\mathcal{L}_\text{reg}$ is a regularization loss that avoids collapsing to a trivial solution of all zeros, and $\lambda$ is a weighting factor for the joint loss for training $g_{enc}$. Noted that for computing either the speech loss $\mathcal{L}_\text{SI-SNR}$ or the speaker loss in $\mathcal{L}_{\tiny\text{CSC}}$, we need to assign the correct target source. During training, we use the utterance-level permutation invariant training (u-PIT) method \cite{yu2017permutation} to solve the permutation issue. We start with using u-PIT to calculate the speech loss $\mathcal{L}_{\tiny\text{SI-SNR}}$, which expends heavy computation at reconstructing every detail of the signals, while often ignoring the global context. We use the separative embedding to modulate the signal reconstruction in the speech-stimuli space by a FiLM \cite{Ethan2017film} method so that the permutation solution for one task also applies to the other task. We switch to using u-PIT to calculate the speaker loss $\mathcal{L}_{\tiny\text{CSC}}$ instead after an empirical number of epochs, thereafter, the speech reconstruction becomes PIT-free and achieves a notable training acceleration.
\vspace{-0.5em}
\section{Experiments}
\vspace{-0.5em}
\label{sec:4}
\subsection{Datasets and Model Setup}
\vspace{-0.5em}
\label{sec:dataset_model}
We firstly evaluated and compared the embedding learnt by \textbf{CPC}\cite{nce18} versus our proposed \textbf{CSC} on a small benchmark dataset WSJ0-2mix\cite{hershey2016deep} for a separation task. The separation SDR result by \textbf{CPC} significantly decreased by an absolute $2.8$dB comparing to \textbf{CSC}, and therefore we continued with \textbf{CSC} to compare with the traditional SV methods on a larger-scale dataset, Libri-2mix, where we split the Librispeech \cite{panayotov2015librispeech} into (1) a training set containing 12055 utterances drawn from 2411 speakers, i.e., $5$ utterances per speaker\footnote{This limitation makes our corpus much more challenging than the lately released LibriMix \cite{cose2020librimix}, meanwhile, more realistic as it was often hard to collect numerous utterances from the user in real industrial applications.}, (2) a validation set containing another $7233$ utterances drawn from the same $2411$ speakers, and (3) a test set containing $4380$ utterances evenly drawn from $73$ speakers. 

To built our proposed model \textbf{CSC}, we inherited from DPRNN's setup \cite{luo2019dual} for the raw waveform encoding, segmentation, and decoding for the SS task. Then, we used $4$, $2$, and $2$ GALR \cite{max2021galr} blocks for $g_{enc}$, $g_{spk}$, and $g_{ss}$, respectively. Inherited from the settings in \cite{max2021galr, luo2019dual}, a 5ms (8-sample) window length and $D=128, K=128, Q=16$ were used. We empirically set $\lambda=10$. For each training epoch, mixture signals of $4$s were generated online by mixing each clean utterance with another random utterance in the training set in a uniform SIR of $0-5$dB.
For testing, mixture signals were pre-mixed in the same SIR range.

For reference systems, we built a SincNet-based SV model \cite{Yoshua2018sv} \textbf{SincNet}, for it was a conventional speaker-vector-based neural network achieving reproducible and strong performance on Librispeech. Moreover, for the ablation study, we built another reference system \textbf{CE} by using $6$ GALR blocks for $g_{spk}$ (note that here $g_{enc}$ could merge into $g_{spk}$ as they no longer had the distinction and that the overall model size kept unchanged for the speaker task) and ablated the proposed method by removing the $g_{att}$ and $g_{ss}$ models from the graph and replacing the proposed loss with the supervised Cross-Entropy (CE) loss.
All the models were implemented and trained using PyTorch.
\vspace{-0.5em}
\subsection{Results and Discussion}
\vspace{-0.5em}
\label{sec:result}
To evaluate the discriminative power of the learned representations, Fig.\ref{fig:roc} shows the Equal Error Rate (EER) and Area Under Curve (AUC) as the performance metrics on the SV task. ``\textbf{[Model]}\_mix", ``\_enh", and ``\_clean" indicate systems training and test on mixture data, enhanced/separated data (by pre-processing the mixture using a SOTA SS system \cite{max2021galr}), and clean data (which we used as a high-bar reference), respectively. Our proposed system outperforms all reference models by a large margin, suggesting our learned representations have strong discriminative power and can achieve high performance in difficult conditions.

Moreover, as shown in Fig. \ref{fig:male_female}-\ref{fig:female_female}, the proposed cross-attention mechanism improves the model transparency, since we can easily interpret the attended content, that at any given time segment, it is generally the most salient source in the mixture that triggers the corresponding cross-attention curve. In places where both sources are soft, all cross-attention curves are low to ignore the possibly noisy and unreliable parts. Also, cross-attention curves rarely raise concurrently. These are all surprisingly similar to a human's auditory selective attention \cite{mesga2012nature, costa2013tunein, obleser2015selective, sullivan2015eeg} in behavioral and cognitive neurosciences, i.e., a listener can not attend to both two concurrent speech streams, while he usually selects the attended speech and ignores other sounds in a complex auditory scene. Note that this selective cross attention property of our model was automatically learned from the mixture without resorting to any regularization or strong supervision.

\section{Conclusions}
\vspace{-0.5em}
\label{sec:conc}
This paper introduces a novel Contrastive Separative Coding method to draw useful representations directly from observations in complex auditory scenarios. It is helpful to learn low-level shared representations towards various levels of separative task. In-depth theoretical studies are provided on the proposed CSC loss regarding the mutual information estimation and maximization, as well as its connection to the existing prior work. The proposed cross-attention mechanism is shown effective in extracting the global aggregation of information across different corrupted observations in various interfering conditions. The learned representation have strong discriminability that its performance is even approaching the clean-condition performance of a conventional fully-supervised SV system.
Another interesting observation is the automatically learned bottom-up cross attentions that are very similar to an auditory selective attention, and we will explore this merit on speaker diarization in our future work.


\bibliographystyle{IEEEbib}
\bibliography{strings,refs}

\end{document}